\documentclass[conference,a4paper]{IEEEtran}
\usepackage{caption}
\usepackage{subcaption}
\usepackage{graphicx}
\usepackage{physics}
\usepackage{amsmath}
\usepackage{amssymb}
\usepackage{amsthm}
\usepackage{enumerate}
\usepackage[shortlabels]{enumitem}
\usepackage{framed}
\usepackage{adjustbox}
\usepackage{multirow}
\usepackage{multicol}
\usepackage{listings}
\usepackage{tikz}
\usepackage{pgfplots}
\usetikzlibrary{decorations.markings,arrows,shapes,chains,matrix,positioning,scopes,calc,quantikz}
\usepackage{xcolor}
\usepackage{bm}
\usepackage{mathtools}
\usepackage{algorithm, float}
\usepackage{algpseudocode,eqparbox}

\usepackage[section]{placeins}
\mathtoolsset{showonlyrefs}
\allowdisplaybreaks
\usepackage{letltxmacro}
\LetLtxMacro{\originaleqref}{\eqref}
\renewcommand{\eqref}{Eq.~\originaleqref}

\definecolor{codegreen}{rgb}{0,0.6,0}
\definecolor{codegray}{rgb}{0.5,0.5,0.5}
\definecolor{codepurple}{rgb}{0.58,0,0.82}
\definecolor{backcolour}{rgb}{0.95,0.95,0.92}

\lstdefinestyle{mystyle}{
    backgroundcolor=\color{backcolour},   
    commentstyle=\color{codegreen},
    keywordstyle=\color{magenta},
    numberstyle=\tiny\color{codegray},
    stringstyle=\color{codepurple},
    basicstyle=\ttfamily\footnotesize,
    breakatwhitespace=false,         
    breaklines=true,                 
    captionpos=b,                    
    keepspaces=true,                 
    numbers=left,                    
    numbersep=5pt,                  
    showspaces=false,                
    showstringspaces=false,
    showtabs=false,                  
    tabsize=2
}

\usepackage{stmaryrd}

\usepackage[style=ieee,backend=biber]{biblatex}
\bibliography{bpqm}

\usepackage{ifthen}

\renewcommand{\triangleq}{\coloneqq}

\newif\ifarxiv
\arxivtrue

\makeatletter
\newcounter{IEEE@bibentries}
\renewcommand\IEEEtriggeratref[1]{%
  \renewbibmacro{finentry}{%
    \stepcounter{IEEE@bibentries}%
    \ifthenelse{\equal{\value{IEEE@bibentries}}{#1}}
    {\finentry\@IEEEtriggercmd}
    {\finentry}%
  }%
}
\makeatother

\usepackage{mathtools}

\usetikzlibrary{shapes.geometric, arrows}
\tikzstyle{startstop} = [rectangle,  minimum width=3cm, minimum height=1cm,text centered,text width=10cm, draw=black ,fill=gray!20]
\tikzstyle{process} = [rectangle, minimum width=3cm, minimum height=1cm, text centered,text width=10cm, draw=black,fill=orange!20]
\tikzstyle{arrow} = [thick,->,>=stealth]
\tikzstyle{state}=[shape=circle,draw=blue!50,fill=blue!20]
\tikzstyle{observation}=[shape=rectangle,draw=orange!50,fill=orange!20]
\tikzstyle{lightedge}=[<-,dotted]
\tikzstyle{mainstate}=[state,thick]
\tikzstyle{mainedge}=[<-,thick]
\usepackage{color}
\definecolor{bitcolor}{rgb}{1,0.84314,0}
\definecolor{checkcolor}{rgb}{0.52941,0.80784,1}
\usepackage{pgfplots}
\pgfplotsset{compat=1.18}
\renewcommand{\epsilon}{\varepsilon}

\newcommand{\vnop}{\varoast}

\newcommand{\cnop}{\boxast}

\newtheorem{theorem}{Theorem}
\newtheorem{prop}[theorem]{Proposition}

\usepackage{authblk}

\IEEEoverridecommandlockouts

\title{Quantum State Compression with Polar Codes\thanks{This research was supported in part by the National Science Foundation (NSF) under Grants 1908730, 1910571, and 2106213. Any opinions, findings, and conclusions or recommendations expressed in this material are those of the author(s) and do not necessarily reflect the views of the NSF.} }

\author[1,2]{Jack Weinberg}
\author[1,3]{Avijit Mandal}
\author[1,2,3]{Henry D. Pfister}
\affil[1]{Department of Electrical and Computer Engineering, Duke University}
\affil[2]{Department of Mathematics, Duke University}
\affil[3]{Duke Quantum Center, Duke University}
\date{}

\begin{document}

\onecolumn

\maketitle

\begin{abstract}
  In the quantum compression scheme proposed by Schumacher, Alice compresses a message that Bob decompresses.
  In that approach, there is some probability of failure and, even when successful, some distortion of the state.
  For sufficiently large blocklengths, both of these imperfections can be made arbitrarily small while achieving a compression rate that asymptotically approaches the source coding bound.
  However, direct implementation of Schumacher compression suffers from poor circuit complexity.

  In this paper, we consider a slightly different approach based on classical syndrome source coding.
  The idea is to use a linear error-correcting code and treat the message to be compressed as an error pattern.
  If the message is a correctable error (i.e., a coset leader) then Alice can use the error-correcting code to convert her message to a corresponding quantum syndrome.
  An implementation of this based on polar codes is described and simulated.
  As in classical source coding based on polar codes, Alice maps the information into the ``frozen" qubits that constitute the syndrome. To decompress, Bob utilizes a quantum version of successive cancellation coding.
\end{abstract}

\section{Introduction}
\label{Introduction}
    
    Quantum computation is the use of quantum mechanical effects for information processing.
    At sufficient scale, quantum computers hold promise for myriad problems including integer factorization, cryptography, and the simulation of physical systems which are intractable on contemporary classical computers \cite[Chapter~1]{nielsen_chuang_2010}.
    As with classical computers, efficient methods for the compression of quantum information will be instrumental for practical quantum computation~\cite{braunstein2000quantum}.
    In 1995, Schumacher proposed the first method for rate-optimal lossless quantum state compression~\cite{schumacher_1995} and it can be seen a generalization of Shannon's original protocol for rate-optimal lossless classical compression.
    However, direct implementation on a quantum computer is quite complex because, for $n$ qubits, it involves rearranging the $2^n$ basis elements in a somewhat complicated fashion~\cite{cleve1996schumacher}.

  One can also use linear error-correcting codes to implement rate-optimal classical compression~\cite{ancheta_1976, cronie_korada_2010}.
    Here, we describe the associated method for rate-optimal quantum compression based on linear error-correcting codes.
    In particular, we use polar codes~\cite{arikan_2009}.
    Polar codes are known to be rate-optimal for many coding and compression problems~\cite{arikan_2009,cronie_korada_2010,wilde_2013,wilde_guha_2013} and, for many of these, they allow efficient encoding and decoding.
    Until recently, however, extensions to quantum problems did not naturally lead to efficient decoding algorithms.
    
    For the pure-state classical-quantum (CQ) channel, belief-propagation with quantum messages (BPQM) can provide optimal decoding for systems defined by tree-like factor graphs~\cite{Renes-njp17,Rengaswamy-npjqi21,piveteau2022quantum,Pfister-bits23}.
    This approach has also been extended to binary-input symmetric CQ channels and polar codes~\cite{brandsen2022belief,mandal2024polar}.
    In this paper, an efficient quantum successive cancellation decoding algorithm is defined for the proposed compression based on polar codes.
    
    We begin with an overview of requisite classical and quantum coding theory. Then, we describe the compression protocol and provide results describing its asymptotic performance.
    Next, we provide simulation results for our protocol with blocklength $8$ and $16$.
    Finally, we conclude with a discussion of this protocol's relationship with, and implications for, quantum compression more generally.

\section{Background}
\label{Background}

\subsection{Binary Linear Codes}
    
    Let $\mathbb{F}_2$ denote the finite field with 2 elements.
    An $[N,K]$ binary linear code $\mathcal{C}$ is a $K$-dimensional subspace of $\mathbb{F}_2^N$.
    Such a code can be defined as the row space of a full-rank generator matrix $G \in \mathbb{F}_2^{K \times N}$ or as the null space of a full-rank parity-check matrix $H \in \mathbb{F}_2^{(N-K)\times N}$.
    The dual code $\mathcal{C}^\perp$ is the orthogonal complement of $\mathcal{C}$, where orthogonality is defined with respect to the standard inner product on $\mathbb{F}_2^N$.
    Thus, a parity-check matrix for $\mathcal{C}$ is a generator matrix for $\mathcal{C}^\perp$.
    
    For an error vector $z \in \mathbb{F}_2^N$, the syndrome of $z$ is given by $s = z H^T$.
    We note that $s = 0$ if and only if $z \in \mathcal{C}$.
    Since $\mathcal{C}$ is a subgroup of the additive group of $\mathbb{F}_2^N$, it follows that $\mathbb{F}_2^N$ can be partitioned into cosets of $\mathcal{C}$ and all elements in a coset will have the same syndrome.
    In each coset, one can choose a coset leader by selecting an element of minimum Hamming weight and breaking ties arbitrarily~\cite{van_1971}.
    
    For communication over a binary symmetric channel (BSC), the transmitted codeword $x \in \mathbb{F}_2^N$ may be corrupted by an error vector $z \in \mathbb{F}_2^N$, resulting in the received bit string $y=x+z$.
    The \emph{optimal decoder} then returns the codeword closest to a codeword with ties broken arbitrarily.
    Such a scheme can also be implemented by syndrome decoding, where the syndrome $s= y H^T$ is computed first and then the error estimate $\hat{z} \in \mathbb{F}_2^n$ is computed from $s$ by selecting the coset leader of the coset associated with syndrome $s$ \cite{ancheta_1976}.


\subsection{Polar Codes}
    \label{Polar Codes}
    Arikan's polar transformation is defined by an invertible matrix $G_N \in \mathbb{F}_2^{N \times N}$ that maps $\mathbb{F}_2^N$ to itself via $x = u G_N$~\cite{arikan_2009}.
    The vector $x$ is transmitted over a memoryless channel with capacity $C$ to give the output $y$.
    For the $i$-th input bit, one can define an effective channel whose input is $u_i$ and whose output is $(y,u_1^{i-1})$.
    As $N$ tends to infinity, the capacities of the individual effective channels become polarized and, for any $\epsilon \in (0,1/2)$, the proportion of ``good" channels (i.e., channels with capacity greater than $1-\epsilon$) converges to $C$ while the proportion of ``bad" channels (i.e., channels with capacities less than $1-\epsilon$) converges to $1-C$.
    
    When communicating via a polar code, one only sends information over the good channels.
    Thus, the $N$ bit message consists of $K$ ``information" bits which constitute a message and $N - K$ ``frozen" bits which are known by the receiver and thus carry no information.
    By freezing the bad channels, one enables the successive cancellation decoder to recover bits transmitted across good channels with high probability.
    
 The polar transformation $G_N$ on $N$ bits may be defined recursively as follows:
    \begin{itemize}
        \setlength{\itemsep}{0.5mm}        
        \item $R_N$ is the $N \times N$ permutation matrix for the reverse shuffle permutation $(1, 3, ... , N - 1, 2, 4, ... , N)$.
        \item $G_N \triangleq (I_{N / 2} \otimes G_2) R_N (I_2 \otimes G_{N / 2})$, where \[ G_2 \triangleq \begin{bmatrix}
            1 & 0 \\
            1 & 1
        \end{bmatrix}. \]
    \end{itemize}
    While originally designed for error correction on classical channels, polar codes have been adapted to many problems.

\subsection{Syndrome Source Coding}\label{ssc}

    Linear codes, such as polar codes, can also be used for lossless compression.
    Syndrome source coding is one such method for adapting linear codes to lossless compression~\cite{ancheta_1976, cronie_korada_2010}.
    In syndrome source coding, one compresses the vector $x\in \mathbb{F}_2^N$ by computing its syndrome $s= x H^T$ using the parity-check matrix $H$ of a linear code.
    This reduces the vector length because $N-K \leq N$.
    To decompress, the decoder maps the syndrome $s$ to the coset leader of its associated coset.
    
    This scheme is successful if and only if the message to be compressed is a coset leader~\cite{ancheta_1976}. Likewise, syndrome decoding of an error is successful if and only if the error is a coset leader.
    Thus, syndrome source coding is successful if and only if the vector to be compressed is an error that is correctable by the syndrome decoder of the code.
    We also note that the successive cancellation decoder for polar codes is easily transformed into a polar syndrome decoder~\cite{cronie_korada_2010} which, for the BSC, corrects the same errors as the original decoder.
    
\subsection{Quantum Formalism}

The set of natural numbers is denoted by $\mathbb{N}=\left\{ 1,2,\ldots\right\} $ and we use the shorthand $[n]\coloneqq\left\{ 1,\ldots,n\right\} $ for $n\in\mathbb{N}$.
We denote the $n$-dimensional Hilbert space by $\mathcal{H}_{n}$. We call a unit length vector $\ket{\psi}\in \mathcal{H}_{n}$ a \emph{pure state}.
An ensemble of $m$ pure states $\{p_i,\ket{\psi_i}\}$ in which $p_{i}$ is the probability of choosing the pure state $\ket{\psi_{i}}$ is represented using a positive semi-definite, unit trace matrix $\rho$ which we call a \emph{density matrix}. In other words, for the ensemble $\{p_i,\ket{\psi_i}\}$, $\rho$ can be written as 
\begin{align}
    \rho=\sum_{i \in [m]}p_{i}\ketbra{\psi_{i}}{\psi_{i}}.
\end{align}
The map $\ket{\psi}\rightarrow U\ket{\psi}$ for unitary $U\in \mathbb{C}^{n\times n}$ is called the \emph{unitary evolution} of the state $\ket{\psi}\in \mathcal{H}_{n}$. Thus, the action of $U$ on the ensemble $\{p_i,\ket{\psi_i}\}$ is described by the density matrix $\tilde{\rho}$, where 
\begin{align}
    \tilde{\rho}= \sum_{i\in [m]} p_{i}U\ketbra{\psi_{i}}{\psi_{i}}U^{\dagger}=U\rho U^{\dagger}
\end{align}
and $U^\dagger$ is the Hermitian transpose of $U$.

An $m$-outcome projective measurement on a state in $\mathcal{H}_{n}$ is implemented through a set of $m$ orthogonal projection matrices $\Pi_{j} \in \mathbb{C}^{n\times n}$ for $1\leq j\leq m$, where $\delta_{i,j}$ is the Kronecker delta function.  These projection matrices satisfy $\Pi_{i}\Pi_{j}=\delta_{i,j}\Pi_{i}$ and $\sum_{j}\Pi_{j}=I$. We denote this measurement by $\hat{\Pi}=\{\Pi_{j}\}|_{j=1}^{m}$. When we apply the measurement $\hat{\Pi}$ on the state $\rho$, the probability of outcome $j$ is $p_j=\text{Tr}(\Pi_{j}\rho)$ and the resulting post-measurement state is $\rho_j=\frac{\Pi_{j}\rho\Pi_{j}}{\text{Tr}(\Pi_{j}\rho)}$. We use the notation $\rho_{A^N}$ to denote the tensor product state of $N$ quantum states in quantum systems $A_{1},\dots A_{N}$ such that
\begin{align}
    \rho_{A^N}=\rho_{A_{1}}\otimes\dots \otimes \rho_{A_{N}}, 
\end{align}
where $\rho_{A_{i}}$ corresponds to quantum state in $A_{i}$. Similarly, we use the shorthand notation $\ket{\psi_{x^N}}_{A^N}$ to denote the tensor product of pure states in system $A_1,\dots, A_{N}$ i.e.
\begin{align}
     \ket{\psi_{x^N}}_{A^N}=\ket{\psi_{x_1}\dots \psi_{x_{N}}}_{A^N}=\ket{\psi_{x_1}}_{A_{1}}\otimes\dots\otimes\ket{\psi_{x_N}}_{A_{N}}.
\end{align}

We also define the embedding of $\mathbb{F}_2^N$ into $\mathbb{C}^{2^N}$ via the mapping $(a_1,a_2,...,a_N) \in \mathbb{F}_2^N$ to $\ket{a_1a_2...a_N}$.
Similarly, given a one-to-one boolean function $f: \mathbb{F}_2^M \to \mathbb{F}_2^N$ with $N\geq M$, we define $E(f)$ to be its embedding into the space of isometries from $A =\mathbb{C}^{2^M}$ to $B=\mathbb{C}^{2^N}$ via
\begin{equation} \label{eq:f_iso}
  E(f) \coloneqq \sum_{x^M\in \mathbb{F}_{2}^{M}} \prescript{}{B}{\ket{f(x^M)}} \bra{x^M}_A.
\end{equation}


\subsection{Quantum Compression by Schumacher}
\vspace{-0mm}


    Schumacher \cite{schumacher_1995} proposed a direct generalization of Shannon's protocol for lossless compression to the quantum setting. To understand Schumacher's protocol, we begin with definitions related to classical compression.
    As usual, the aim of the compression protocol is for Alice to compress a message which Bob subsequently decompresses. 

    Consider $N$-bit strings produced by a classical source with alphabet $\mathcal{X}$, alphabet size $|\mathcal{X}|=n$ and distribution $p_{X}$. The entropy $H(X)$ of such a source is defined as 
    \begin{align}
        H(X) \triangleq - \sum_{x \in \mathcal{X}} p_X(x) \log p_X(x).
    \end{align}
    The sample entropy $\overline{H}(x^N)$ of a sequence $x^N$ is defined as 
    \begin{align}
        \overline{H}(x^N) \triangleq - \frac{1}{N} \log(p_{X^N} (x^N)).
    \end{align}
    The $\delta$-typical set $T_\delta^{X^N}$ is defined as 
    \begin{align}
        T_\delta^{X^N} \triangleq \{x^N :|\overline{H}
        (x^N) - H(X)| < \delta\}.
    \end{align}
    With these classical definitions, we are equipped to understand the quantum problem.

    Consider a \emph{quantum information source} described by the state $\rho$. We can express the state $\rho$ using the eigenvalue decomposition as 
    \begin{align}\label{rho-eigen}
        \rho=\sum_{x\in \mathcal{X}}p_{X}(x)\ketbra{\psi_{x}}{\psi_{x}}
    \end{align}
    where $p_{X}(x)$ corresponds to the probability of choosing the pure state $\ket{\psi_{x}}$ and $\braket{\psi_{x}}{\psi_{x'}}=0$ if $x\neq x'$. The von Neumann of $\rho$ is
    \begin{align}
        S(\rho)=-\text{Tr}(\rho \log \rho)=-\sum_{x\in \mathcal{X}}p_{X}(x)\log (p_{X}(x)).
    \end{align}
    In other words, the state $\rho$ can be represented by the ensemble $\{p_{X}(x),\ket{\psi_x}\}$.
    $\rho$ can be decomposed in many ways using non-orthogonal states. However, the representation of state $\rho$ in terms of eigenvectors achieves von Neumann entropy which ensures maximum compressibility. Similarly, consider joint quantum state $\rho_{A^N}$ with $N$ states drawn from the ensemble  $\{p_{X}(x),\ket{\psi_x}\}$. Using the representation of $\rho$ in~\eqref{rho-eigen}, we can decompose $\rho_{A^N}$ in the following form
    \begin{align}
      \rho_{A^N}=\sum_{x^N\in \mathcal{X}^N}p_{X^N}(x^N)\ketbra{\psi_{x_1}\dots\psi_{x_N}}{\psi_{x_1}\dots\psi_{x_N}}_{A^N} 
    \end{align}
    where $p_{X^N}(x^N)=p_{X}(x_1)\dots p_{X}(x_N)$.  The classical sequences $x^N\in \mathcal{X}^N$ correspond to the indices of the quantum state in the decomposition of $\rho_{A^N}$. To define quantum typicality, we construct the typical subspace as follows 
    \begin{align}
      \Pi_{A^N}^{\rho,\delta} &\ =\text{span}\{\ket{\psi_{x^N}}: x^N\in  T_\delta^{X^N}\}\\
      &\ = \sum_{x^N\in T_\delta^{X^N}}\ketbra{\psi_{x^N}}{\psi_{x^N}}_{A^N}.
    \end{align}
    This is exactly the subspace spanned by states $\ket{\psi_{x^N}}$ whose labels $x^N\in \mathcal{X}^N$ are $\delta$-typical with respect to the distribution $p_{X^N}(x^N)$. In other words, the notion of quantum typical is similar to classical typicality but quantum typicality is considered in the eigenbasis. 
    Schumacher compression exploits this notion of quantum typicality and the typical subspace corresponding to the state $\rho_{A^N}$ to achieve the quantum compression limit i.e. the von Neumann entropy $S(\rho)$ \cite{wilde_2013}.
Consider the projective measurement defined by \[ \{  \Pi_{A^N}^{\rho,\delta}, I -  \Pi_{A^N}^{\rho,\delta} \}. \]
    If the first outcome occurs, then the quantum state is projected onto a typical subspace.
    If the second outcome occurs, it is projected onto the orthogonal complement and failure is declared.
    With the addition of a flag qubit $F$, the action of this measurement on the density matrix $\rho_{A^N}$ is given by
    \begin{align}
      \rho_{A^N} \mapsto &\ (I \!-\! \Pi_{A^N}^{\rho,\delta}) \rho_{A^N} (I \!-\! \Pi_{A^N}^{\rho,\delta}) \!\otimes\! \ketbra{0}_F + \\
      &\ \quad \Pi_{A^N}^{\rho,\delta}  \rho_{A^N} \Pi_{A^N}^{\rho,\delta} \!\otimes\! \ketbra{1}_F.
    \end{align}
    This measurement is called the typical subspace projection as it projects to the $\delta$-typical subspace of $A^N$.
    In general, this also causes a small distortion of the state because a small fraction state's mass lies outside of the typical set. More precisely, it follows from $\Tr(\Pi_{A^N}^{\rho,\delta}) < 2^N$ that $\|\rho_{A^N} - \varepsilon(\rho_{A^N})\| > 0$ if $\rho$ is non-singular, where $\varepsilon$ is the channel defined by the typical subspace projection. However, this error vanishes as $N \to \infty$ provided that $\delta > 0$ and $S(\rho) < 1$. 
    After the projection, Alice's task is to compress the projected state $\Pi_{A^N}^{\rho,\delta}  \rho_{A^N} \Pi_{A^N}^{\rho,\delta}$ before sending it to Bob.
    Since
    \begin{align}
        \text{Tr}\left(\Pi_{A^N}^{\rho,\delta}\right) = \left| \Pi_{A^N}^{\rho,\delta} \right| \leq 2^{N(S(\rho)+\delta)},
    \end{align}
    there is a bijective boolean function $f : T_\delta^{X^N} \xrightarrow{} \{0,1\}^{N[S(\rho) + \delta]}$ is a bijection from the classical typical sequences to the set of binary sequences of length $N(S(\rho)+\delta)$.
    Thus, Alice can use \eqref{eq:f_iso} to construct the isometry $U_f = E(f)$ mapping $A^N$ to $\mathbb{C}^{2^{N[S(\rho) + \delta]}}$.
    Observe that this isometry is exactly the map from the set of message strings to the set of typical strings used in Shannon's protocol.
    To decompress Alice's state, Bob applies $U_f^{\dagger}$ on the received state. This strategy approaches the fundamental limiting compression rate of $S(\rho)$ . 

\subsection{Quantum Compression via Syndrome Source Coding}

Linear codes may also be adapted to implement lossless quantum state compression.
We describe this adaptation here in a generic way before later addressing the particular case of polar codes, which is the main subject of this paper. As described in section \ref{ssc}, any linear code can be used to implement lossless classical compression via syndrome source coding. Let $\mathcal{C}$ be an $[N,K]$ binary linear code with full-rank parity-check matrix $H$.
Let $H'$ be an invertible extension of $H$ that is formed by adding $K$ rows to the bottom of $H$ so that $H'$ is full rank.

Suppose Alice wants to compress a state drawn a quantum information source described by $\rho^{\otimes N}$, where $\rho$ describes a single qubit state with spectral decomposition
\begin{align}
    \rho=(1-p)\ketbra{\psi_{0}}{\psi_0}+p\ketbra{\psi_1}{\psi_1}.
\end{align}
Then, Alice and Bob may compress and decompress $\rho^{\otimes N}$ by embedding the syndrome-source coding protocol for $\mathcal{C}$ into the quantum domain as follows:
\begin{enumerate}
    \setlength{\itemsep}{0.5mm}
    \item Alice applies $U_\rho^{\otimes N}$ to her state, where $U_\rho$ is the unitary that diagonalizes $\rho$, to get
    \begin{align}
    \overline{\rho} = U_{\rho}\rho U_{\rho}^{\dagger}=(1-p)\ketbra{0}{0}+p\ketbra{1}{1}.
    \end{align} 
    \item Alice applies the quantum instrument map $\varepsilon$ defined by $\overline{\rho} \mapsto (\Pi^N_K \overline{\rho} \Pi^N_K) \otimes \ket{1} \bra{1}_B + ((I_N - \Pi^N_K) \overline{\rho} (I_N - \Pi^N_K)) \otimes \ket{0} \bra{0}_B$ to her state, where $\Pi^N_K \triangleq \sum_{x^N \in T} \ket{x^N} \bra{x^N}$ and $T$ is the set of computational basis states indexed by coset leaders of $\mathcal{C}$ (i.e., the set of errors correctable by syndrome decoding).
    The state of the resulting system is given by $\Tilde{\rho}^{\otimes N}$.
    \item Alice measures system $B$ with the projective measurement $\{\ketbra{0}{0}, \ketbra{1}{1}\}$.
    If the outcome is $\ketbra{0}$, failure is declared.  Otherwise, the outcome is $\ketbra{1}$ and she applies the unitary embedding of $H'$, represented by $U_{H'} = E(H' x)$, to $\tilde{\rho}$. We use the extended matrix $H'$ because the matrix $H$ is typically not invertible and thus cannot be embedded into a unitary.
    Then, Alice sends the first $N-K$ qubits to Bob.
    From a density matrix point of view, this is equivalent to computing the partial trace over the system $\mathcal{I}$ which consists of the last $K$ qubits. Thus, Bob receives the state $\Psi = \Tr_\mathcal{I}(U_{H'} \rho U_{H'}^\dagger)$, which is the mapping into syndrome space of the projection of $\rho$ onto the correctable errors in $T$.
    \item To decompress $\Psi$, Bob implements syndrome decoding as an isometry.
    Consider the isometry mapping $N-K$ qubits to $N$ qubits, that is defined by \[U_D = \sum_{x^N\in T} \prescript{}{A^N}{\ket{x^N}}\bra{H x^N}. \]
    Applying $U_D$ to $\Psi$ gives $\tilde{\rho}$ because the image of $\Pi_K^N$ is supported on $T$ by construction and thus $U_D$ inverts the both the partial trace and the syndrome mapping.
    Finally, Bob applies ${U_\rho^\dagger}^{\otimes N}$ to invert the diagonalization operation applied by Alice. This gives the desired approximation of Alice's initial state.
\end{enumerate}

\section{Protocol}
\label{Protocol}
Suppose Alice wants to store an $N$-qubit tensor product state using as few qubits as possible; she intends to give her stored state to Bob who will then decompress it to recover Alice's original state. Alice's message may be understood as being drawn from a quantum information source described by the state $\rho^{\otimes N}$ where $\rho$ describes a single qubit state. Suppose that $\rho$ has a spectral decomposition
\begin{align}
    \rho=(1-p)\ketbra{\psi_{0}}{\psi_0}+p\ketbra{\psi_1}{\psi_1},
\end{align}
where $\psi_{0},\psi_1\in \mathcal{H}_{2}$ are eigenvectors of state $\rho$. Subsequently, we refer to the probabilities $\{p,1-p\}$ as probabilities of source qubits for the state $\rho$. We also assume that Alice has access to the unitary $U_{\rho}\in \mathbb{C}^{2\times 2}$ such that diagonalizes $\rho$.
This assumption is valid because Alice knows the state $\rho$ and we assume that Bob also knows the unitary $U_{\rho}$, which he uses while recovering the compressed state.

Below, we describe how polar codes can be used to implement lossless quantum compression, which we call Schumacher compression in a generic sense. Since $U_{\rho}$ corresponds to a qubit unitary, both Alice and Bob can inexpensively apply it to their states to realize our compression protocol.

    We propose that Alice encode her state by embedding the syndrome source coding procedure into the group of unitary operators on her state space.
    Following the discussion in section~\ref{Polar Codes}, we can embed the binary polar transform $G_N$ into a unitary transform $V_N$ on $N$ qubits (c.f., \cite{wilde2012quantum}).
    While $G_N$ is designed to act on binary vectors via right multiplication (i.e., $u \mapsto u G_N$), we define $V_N$ act on qubits via left multiplication.
    Thus, we have $V_N = E(G_N^T)$ and this gives
    \begin{itemize}
        \setlength{\itemsep}{0.5mm}
        \item $V_2 \triangleq E(G_2^T (x_1,x_2)^T) = \textrm{CNOT}_{2\to 1}$.
        \item $U^{R}_N \triangleq E(R_N^T x^N) = E(R_N x^N)$ is a SWAP operator on qubits defined by the permutation $(1, 3, ..., N - 1, 2, 4, ..., N)$.
        \item $V_N \triangleq E(G_N^T x^N) = E( (I_{N / 2} \otimes G_2^T) R_N (I_2 \otimes G_{N / 2}^T) x^N)$ which implies $V_N \triangleq (\mathbb{I}_{N/2} \otimes V_{N / 2}) \: U^R_{N} \: (\mathbb{I}_{N-2} \otimes V_2)$.
    \end{itemize}
    where $\mathbb{I}_M$ denotes the identity operator on $M$ qubits.
    We propose the following compression protocol:
    \begin{enumerate}
        \setlength{\itemsep}{0.5mm}
        \item Alice and Bob design an $N$-bit classical polar code for the BSC with error probability $p$, or  BSC$(p)$. They agree on the set $\mathcal{I}$ of indices for the $K$ information qubits so that $\mathcal{I}^c$ contains the indices of the $N-K$ frozen bits.
        Let $f\colon \mathbb{F}_2^{N-K} \to \mathbb{F}_2^N$ be the boolean function defined by polar syndrome decoding that maps syndromes to error patterns.
        Let $T$ be the range of $f$ (i.e., the set of correctable error patterns for polar syndrome decoding).
        
        \item Alice applies a unitary $U_{\rho}^{\otimes N}$ to $\rho^{\otimes N}$ where $U_\rho$ is a change of basis operator from the eigenbasis of $\rho$ to the computational basis: she obtains the state $\overline{\rho} = U_\rho^{\otimes N} \rho^{\otimes N} {U_\rho^{\otimes N}}^\dagger$.
        \item Alice encodes the message:
        \begin{enumerate}[(a)]
      
            \item Alice must apply the quantum instrument map $\varepsilon$ defined by $\overline{\rho} \mapsto (\Pi^N_K \overline{\rho} \Pi^N_K) \otimes \ket{1} \bra{1}_B + ((I_N - \Pi^N_K) \overline{\rho} (I_N - \Pi^N_K)) \otimes \ket{0} \bra{0}_B$ to her state, where $\Pi^N_K \triangleq \sum_{x^N \in T} \ket{x^N} \bra{x^N}$ is the projection onto the set of correctable error patterns, obtaining $\Tilde{\rho}^{\otimes N}$.
            However, directly applying this isometry is intractable for large block lengths. Thus, Alice instead makes use of the efficient quantum successive cancellation decoding algorithm described in Section~\ref{sec:eff_imp} to apply this projection.
            \item Alice measures system $B$ with the projective measurement $\{\ketbra{0}{0}, \ketbra{1}{1}\}$.
            If the outcome is $\ketbra{0}$, failure is declared.  Otherwise, the outcome is $\ketbra{1}$ and she applies $V_N$ to system $A$.
            \item Alice sends the frozen qubits with indices in $\mathcal{I}^c$ to Bob.
            Thus, Bob receives the quantum state $\Tr_{\mathcal{I}} (\Psi)$.
        \end{enumerate}
        \item To decode, Bob must apply the isometry $U_D = E(f)$ to the received qubits.
        However, directly applying this isometry is intractable for large block lengths. Thus, Bob instead makes use of the efficient quantum successive cancellation decoding algorithm described in Section~\ref{sec:eff_imp} to decompress Alice's state.  
        Lastly, Bob applies ${U_\rho^\dagger}^{\otimes N}$ to invert the diagonalization operation Alice applied, thus obtaining an approximation of Alice's initial state.
    \end{enumerate}
    \begin{figure*}
    \begin{center}
    \resizebox{\textwidth}{!}{%
    \begin{tikzpicture}
    \node[]{
    \begin{quantikz}[]
    \lstick[1]{Flag Qubit} & \lstick{$F$} &\qw& \qw & \gate[5]{\varepsilon} & \meter{}  \\
    \lstick[4]{Initial State} & \lstick{$\rho$} &\qw& \gate[1]{U_\rho} \qw& \qw & \qw & \ctrl{1} & \qw & \ctrl{1} & \trash{\text{trace}} \\
    & \lstick{$\rho$} &\qw& \gate[1]{U_\rho} & \qw & \qw& \targ{} & \swap{1} & \targ{} & \trash{\text{trace}} \\
    & \lstick{$\rho$} &\qw& \gate[1]{U_\rho} & \qw &\qw & \ctrl{1} & \swap{-1} & \ctrl{1} & \qw \slice{Transmission} &\qw & \qw & \gate[4]{\mathcal{D}} & \gate[1]{U_\rho^\dagger} &\qw& \rstick[4]{Recovered State} \\
    & \lstick{$\rho$} &\qw& \gate[1]{U_\rho} & \qw & \qw & \targ{1} & \qw & \targ{} & \qw &\qw& \qw & & \gate[1]{U_\rho^\dagger} &\qw& \\
    &&&&&&&&&& & \qw &\qw& \gate[1]{U_\rho^\dagger} &\qw& \\
    &&&&&&&&&& & \qw &\qw& \gate[1]{U_\rho^\dagger} &\qw& 
    \end{quantikz}
    };
    \end{tikzpicture}
    }
    \end{center}
    \caption{Quantum compression with embedded polar codes for blocklength 4.}
    \end{figure*}
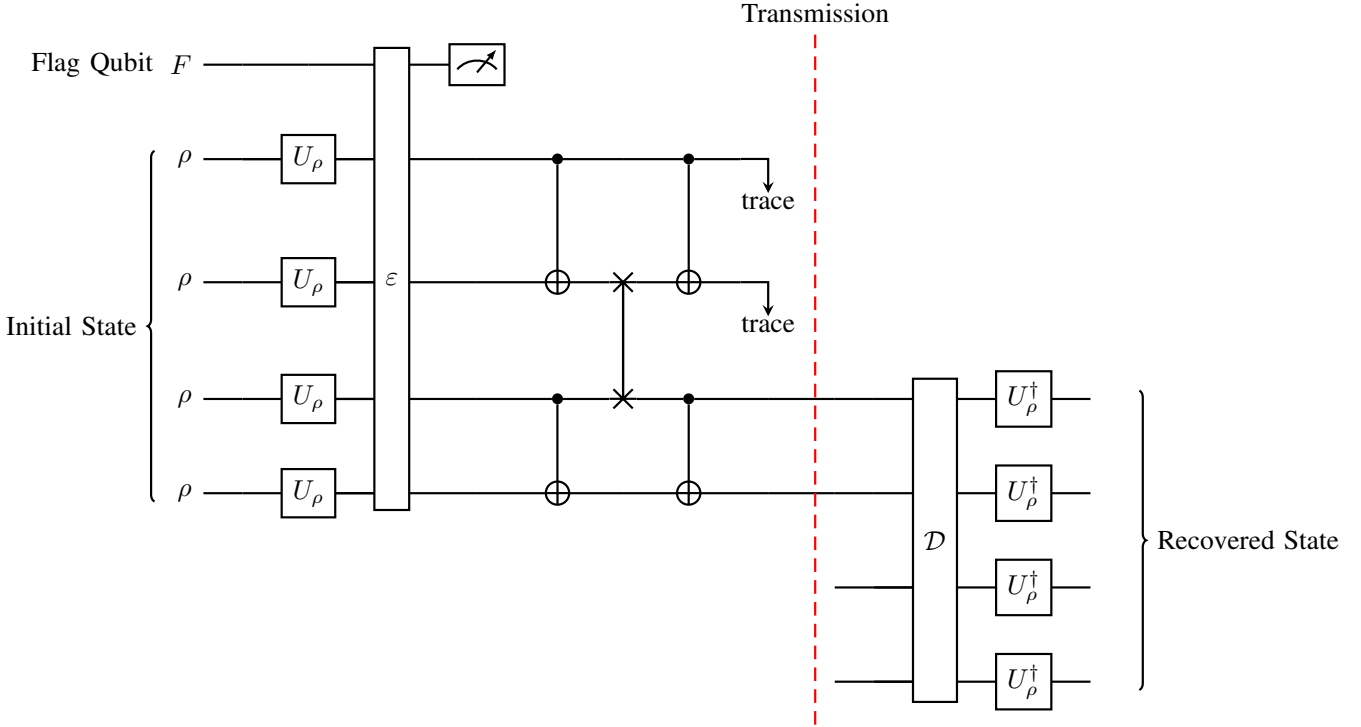
    Whereas Schumacher compression, as it was initially proposed, involves projecting Alice's state onto the subspace of $\delta$-typical strings, in the new protocol, Alice projects onto the subspace of errors which are correctable with respect to a polar code. In the next section, we will discuss how Bob can utilize efficient quantum successive cancellation decoding to realize the required isometry in practice.
\section{Efficient Implementation}
\label{sec:eff_imp}
Belief propagation (BP) is the name used for a class of message-passing algorithms that provide low-complexity decoding for the codes represented by factor graphs~\cite{Kschischang-it01}.
Such codes include low-density parity check (LDPC) \cite{gallager1962low} and polar \cite{arikan_2009} codes.
While some versions appeared earlier, BP was named by Pearl in 1982 \cite{pearl2022reverend} and was later shown to be efficient for decoding codes \cite{mceliece1998turbo,kschischang1998iterative,davey1998low}. Recently, BP has been generalized via belief propagation with quantum messages (BPQM) for pure-state CQ channels  \cite{piveteau2022quantum,Pfister-bits23,Renes-njp17,Rengaswamy-npjqi21} and subsequently for general binary symmetric CQ (BSCQ) channels \cite{brandsen2022belief,mandal2023belief,mandal2024polar}.

In this paper, we exploit the factor graph structure of the polar code to recover the compressed quantum state using a BPQM-type algorithm.  In our method, the decoder receives frozen qubits constituting the compressed quantum state which are used to recover information qubits via successive cancellation decoding of a polar code \cite{arikan_2009}.

In classical BP for polar code, while decoding a noisy codeword the decoder utilizes channel error probabilities as inputs. It combines these messages across the factor graph, employing successive cancellation decoding to determine the error pattern. Similarly, BP can be used for compression where we recover the complete binary sequence/codeword from the frozen bits. In this scenario, the frozen bits are used as a syndrome i.e. a compressed sequence and instead of channel error probabilities, the decoder uses source probabilities as input. The syndrome-based algorithm can be lifted to the quantum domain by considering all possible messages sent back from frozen qubits in superposition when frozen qubits are used as a syndrome. This allows for the combining of messages classically while accounting for the indeterminate nature of quantum information. This observation plays a key role in our lifted BP algorithm. While the classical messages, when passed through the factor graph, become conditional on nature depending on the number of frozen qubits and are written as a list of probabilities, the quantum part of the algorithm involves designing appropriate unitary based on these conditional classical messages.

Initially, the information qubits are prepared as ancilla qubits in state $\ket{0}$. To realize the successive cancellation for decoding bit $u_{i}$, we construct a factor graph with root node $u_{i}$ and decoder output for past bits as $\hat{u}_{1}^{i-1}$ from the decoding factor graph of the polar code. The factor graph takes message probabilities associated with the distribution of the source qubits after diagonalization as inputs.
The messages are passed through the factor graph through check-node ($\cnop$) and bit-node ($\vnop$) combining rules. In the probability domain, the check-node and bit-node combining rules of messages $p_1$ and $p_2$ are realized as  \vspace{-1.5mm}
\begin{align}
    p_1\cnop p_2 &\ = p_1(1-p_2)+p_2(1-p_1)\label{cnop_msg}\\
    p_1\vnop p_2 &\ =\frac{p_1p_2}{p_1p_2+(1-p_1)(1-p_2)}\label{vnop_msg}.
\end{align}
Since the frozen qubits are in a superposition of $\ket{0}$, $\ket{1}$ qubit states, while decoding the root nodes corresponding to the frozen qubits, a conditional message is sent back to continue the successive cancellation decoding. To decode root nodes associated with information qubits, we construct an appropriate unitary based on conditional messages sent back from the frozen qubits and apply the conditional unitary to decide whether or not to flip the information qubit (which is set as an ancilla in state $\ket{0}$). The decoding complexity grows as the number of frozen qubits increases.
\subsection{Decoding a length-4 polar code}
Before describing the general algorithm to recover the quantum state for the length-$N$ polar code, we will describe the decoding associated with the length-4 polar code to introduce the readers to the flavour of the BP for the quantum state compression problem using polar codes. In this toy example, we consider the index of frozen qubits to be $u_{F}=\{1,3\}$ and information qubits to be $u_{I}=\{2,4\}$. The conditional messages are described using two lists $\{p_{vec}, L_{F}\}$ where $L_{F}$ is a length-$l$ sorted list of frozen qubits the message is conditioned on and $p_{vec}$ is the length-$2^l$ list of message probabilities. To denote the $i$\textsuperscript{th} element of the list $L$ we use the notation $L(i)$.
Fig.~\ref{polar4dec1} and Fig.~\ref{polar4dec3} shown in the appendix depict factor graphs along with belief propagation for a length-4 polar code.

The decoder takes message probabilities for the codeword bits $x_{1}^{4}$ as the input. These message probabilities correspond to the distribution of source qubits after diagonalization; the probabilities of $0$ and $1$ are $1-p$ and $p$, respectively. At the beginning of the message passing procedure, the messages do not depend on the frozen qubits, hence each codebit $x_{i}$ contains message $\{\{p\},\{\}\}$. The messages are passed through the factor graph and combined according to check-node and bit-node combining rules. After check-node combining, the nodes $v_1$ and $v_3$ contain messages $\{\{p\cnop p\},\{\}\}$. Then we apply check-node combining again to compute the message for $u_1$ which is $\{\{p_1\cnop p_1\},\{\}\}$ where $p_1=p\cnop p$. Since the first qubit is frozen, it sends back the conditional message $\{\{0,1\},\{1\}\}$ which is combined with the messages in nodes $v_1$ and $v_3$ to continue the successive cancellation decoding. The message from the frozen qubit $u_1$ is combined with the message of $v_1$ to produce the message $\{\{0\cnop p_1,1\cnop p_1\},\{1\}\}$=$\{\{p_1,1-p_1\},\{1\}\}$. This message is combined via a bit-node with the message of node $v_3$ to generate the message to the bit $u_3$ which can be computed as $\{\{p_1\vnop p_1,(1-p_1)\vnop p_1\},\{1\}\}$. This message is used to decide whether or not to flip the information qubit conditioned on the first qubit.
The classical message passing applies hard decisions on the received message and estimates the information bit. Similarly, in our lifted algorithm, we apply the hard decision to get an associated conditional message $\{\{h_0,h_1\},\{1\}\}$ where $h_0=\text{hard-dec}(p_1\vnop p_1)$ and $h_1=\text{hard-dec}((1-p_1)\vnop p_1)$. 
To conditionally flip the information qubit we apply a unitary of the form
\begin{align}
    U_{2,\{1\}}(h_0,h_1)=\ketbra{0}{0}\otimes \sigma_{x}^{h_0}+\ketbra{1}{1}\otimes \sigma_{x}^{h_1}.
\end{align}
on the first and second qubits.
This unitary generalizes the successive cancellation decoding as it coherently flips the information qubit depending on the state of the frozen qubit. In the next stage of the decoding procedure, the conditional message $\{\{h_0,h_1\},\{1\}\}$ associated with the hard decision is sent back from the information qubit corresponding to $u_2$ and passed through the factor graph associated with root-node $u_3$ (see Fig.~\ref{polar4dec3}). Since the qubit corresponding to $u_3$ is a frozen qubit it sends back $\{\{0,1\},\{3\}\}$. Let the conditional messages in node $v_2$ and $v_4$ after message passing at this stage be $\{\{p_{v_2},p_{v_2}'\},\{1\}\}$  and  $\{\{p_{v_4},p_{v_4}'\},\{1\}\}$, respectively. To decode the information qubit for $u_4$, the messages $\{\{0,1\},\{3\}\}$ and $\{\{p_{v_2},p_{v_2}'\},\{1\}\}$ are combined via a check-node to produce conditional message $\{\{p_{v_2},1-p_{v_2},p_{v_2}',1-p_{v_2}'\},\{1,3\}\}$. This message is combined via a bit-node with the message  $\{\{p_{v_4},p_{v_4}'\},\{1\}\}$ to produce $\{\{p_{v_2}\vnop p_{v_4},(1-p_{v_2})\vnop p_{v_4},p_{v_2}'\vnop p_{v_4}',(1-p_{v_2}')\vnop p_{v_4}'\},\{1,3\}\}$. Let the hard decision of the above message be $\{\{h_0',h_1',h_2',h_3'\},\{1,3\}\}$. The decoder then applies a unitary given by
\begin{align}
     U_{4,\{1,3\}}(h_0',h_1',h_2',h_3'&) 
      =\ \ketbra{00}{00}\otimes \sigma_{x}^{h_0'}+\ketbra{01}{01}\otimes \sigma_{x}^{h_1'} \\
     &\ +\ketbra{10}{10}\otimes \sigma_{x}^{h_2'}+\ketbra{11}{11}\otimes \sigma_{x}^{h_3'}.
     \vspace{-3mm}
\end{align}
on the information qubit for $u_4$ and previous frozen qubits.
The unitary conditionally flips the information qubits based on the state of the frozen qubits. This concludes the description of our BP decoding algorithm for length-4 polar code.

A similar algorithm also can be used to realize the projection in the encoding map that Alice uses to compress her state to frozen qubits.
In this case, we use the same factor graphs, but instead of flipping, we apply conditional projectors on the information qubits. The conditional projectors use hard decision messages and depending on the hard decision message Alice determines whether the state belongs to the typical subspace or not. The conditional projector for the information qubit at index 4 conditioned on frozen qubits $\{1,3\}$ with hard-decision message $\{\{h_0',h_1',h_2',h_3'\},\{1,3\}\}$, is of the form 
\begin{align}
    \Pi_{4,\{1,3\}}& (h_0',h_1',h_2',h_3')=  \ketbra{00}{00}\otimes \Pi_0^{1-h_0'}\Pi_{1}^{h_0'} \\
    &+\ketbra{01}{01}\otimes \Pi_0^{1-h_1'}\Pi_{1}^{h_1'} 
    +\ketbra{10}{10}\otimes \Pi_0^{1-h_2'}\Pi_{1}^{h_2'}
    \\ & +\ketbra{11}{11}\otimes \Pi_0^{1-h_3'}\Pi_{1}^{h_3'}.
\end{align}
where $\Pi_{0}=\ketbra{0}{0}$ and $\Pi_1=\ketbra{1}{1}$ respectively.
After determining whether the state belongs to the typical subspace or not, Alice traces out the information qubits and sends the state to Bob. Since we are dealing with computational bases and projectors commuting with each other, Alice can compress the state optimally.
The description of BP for length-4 polar code builds up the intuition about the general recursive algorithm for the length-$N$ which we describe in Appendix~\ref{BP decoding}.

\section{Main Results}
\label{Main Results}
    \begin{prop} \label{prop:error}
        The proposed protocol is successful if and only if Alice measures the outcome associated with projector $\Pi^N_K$, which occurs with probability $\Tr(\Pi^N_K \rho^{\otimes N})$.
    \end{prop}
    \begin{proof}
        This follows from the problem setup, the description in Section~\ref{Protocol}, and the measurement postulate.
    \end{proof}
    
    \begin{prop}\label{asymptotically successful}
    For any $\delta>0$ and $\nu \in \mathbb{N}$, let $N = 2^\nu$ and consider the compression protocol described in Section~\ref{Protocol} for a qubit source with entropy $H = - \Tr(\rho \ln \rho)$ and $K= \lfloor( 1-H-\delta) N \rfloor$.
    Then, the failure probability $P_\nu = 1 - \Tr(\Pi^N_{K} \rho^{\otimes N})$ satisfies $\lim_{n \to \infty} P_\nu = 0$.
    \end{prop}
    \begin{proof}

    Let $\rho^{\otimes N}$ have a spectral decomposition $\rho^{\otimes N} = \sum_{x^N \in \mathbb{F}_2^N} p_X(x^N) \ket{x^N} \bra{x^N}$. 
    For the projector $\Pi^N_K$ defined in Section~\ref{Protocol}, we have 
    \begin{align}
            \Pi^N_K \rho^{\otimes N} &= \sum_{x^N \in T} \sum_{y^N \in \mathbb{F}_2^N} p_{X^N} (y^N) \ket{x^N} \bra{x^N}  \ket{y^N} \bra{y^N} \\ &= \sum_{x^N \in T} \sum_{y^N\in \mathbb{F}_2^N} \delta_{x^N,y^N} \,p_{X^N} (y^N) \ket{x^N} \bra{y^N} \\ &= \sum_{x^N \in T} p_{X^N} (x^N) \ket{x^N} \bra{x^N},
    \end{align}
    where $T$ is the set of correctable errors for the polar syndrome decoder.
    Thus, we find that
    \begin{equation}
        \Tr(\Pi^N_K \rho^{\otimes N}) = \sum_{x^N \in T} p_{X^N} (x^N) \label{eq:succ_prob}
    \end{equation}
    and, by Proposition~\ref{prop:error}, this equals the probability of success for the quantum compression scheme.    
    Without loss of generality, suppose $\rho$ has a spectral decomposition
    \begin{align}
        \rho = (1 - p) \ket{0} \bra{0} + p\ket{1} \bra{1}, 
    \end{align}
    where $p \in [0, 1/2]$.
    Then, $p_{X^N} (x^N)$ equals the distribution of error sequences for the BSC$(p)$ and~\eqref{eq:succ_prob} also equals the probability of decoding success for the polar syndrome decoder on a BSC$(p)$.
    Moreover, it follows that $- \Tr(\rho \ln \rho) = h(p)$.
   
    Given the performance equivalence between polar syndrome decoding on the BSC$(p)$ and the proposed quantum compression scheme, the remainder of our analysis is classical.
    First, we recall that the capacity of a BSC$(p)$ is $C = 1 - h(p)$, where $h\colon [0,1] \to [0,1]$ is the binary entropy function.
    Next, we observe that the code rate of the polar code is given by
    \[ K/N \leq 1-h(p)-\delta = C-\delta. \]
    Sequences of polar codes with rate converging to $C-\delta$ achieve vanishing block error rate on the BSC~\cite{arikan_2009}.
    Thus, it follows that probability of success $\Tr(\Pi^N_K \rho^{\otimes N})$ for the quantum compression protocol satisfies $P_\nu = 1-\Tr(\Pi^N_K \rho^{\otimes N})  \to 0$.
    Likewise, the proposed compression protocol sends $N-K = N(1-K/N) \leq N(h(p)+\delta)$ qubits which can be made arbitrarily close to the Schumacher limit of $N h(p)$ qubits.
    \end{proof}
%

\section{Numerical Results}
\label{Numerical Results}
    We plot the probability of successful quantum state compression with our protocol, as a function of intrinsic bit-flip probability for various design error probabilities, for blocklengths 8 and 16 (Figures \ref{fig:typicality8} and \ref{fig:typicality16}). We also plot the performance of Schumacher compression on the same blocklengths (Figures \ref{fig:schumacher8} and \ref{fig:schumacher16}) and compare the two protocols (Figures \ref{fig:combined8} and \ref{fig:combined16}). The quantum polar codes used here are embedded classical polar codes selected to minimize error rates (as determined via Monte Carlo simulation).
    For blocklengths 8 and 16, our protocol performs as we would expect: Figures \ref{fig:typicality8} and Figures \ref{fig:typicality16} show that the probability of successful compression strictly decreases as entropy increases and the compression rates of codes designed via Monte Carlo simulation strictly increase as design error probability increases.
    Schumacher compression, on the other hand, does not perform in like manner. Figures \ref{fig:schumacher8} and \ref{fig:schumacher16} show no clear relationship between source entropy and compression performance. We attribute this to the short blocklengths, dictated by computing power, on which we simulate Schumacher compression. For larger blocklengths, we expect this protocol to achieve higher accuracy and lower compression rates as source entropy decreases.
    Last, Figures \ref{fig:combined8} and \ref{fig:combined16} demonstrate that our protocol achieves superior succss probabilities and compression rates, as compared to Schumacher compression, for blocklengths 8 and 16. Moreover, changing the $\delta$ value used in Schumacher compression doesn't change this discrepancy. Evidently, for short blocklengths, our protocol's notion of typicality corresponds more closely to the set of likely messages than $\delta$-typicality does, though these notions are equivalent in the limit of large block lengths in the sense that both compression protocol's are capacity-achieving.

    \begin{figure}
        \centering
        \begin{subfigure}[b]{0.45\textwidth}
            \centering
            \input{Tex_graphs/combined_blocklength8}
            \caption[width=\textwidth]{Probability of success as a function of binary entropy bit-flip probability for blocklength 8. Horizontal lines indicate compression rates.}
            \label{fig:typicality8}
        \end{subfigure}
        \hfill
        \begin{subfigure}[b]{0.45\textwidth}
            \centering
            \input{Tex_graphs/typicality_blocklength16}
            \caption[width=\textwidth]{Probability of success as a function of binary entropy bit-flip probability for blocklength 16. Horizontal lines indicate compression rates.}
            \label{fig:typicality16}
        \end{subfigure}
        \\[4ex]
        \begin{subfigure}[b]{0.45\textwidth}
            \centering
\begin{tikzpicture}

\definecolor{darkgray176}{RGB}{176,176,176}
\definecolor{darkorange25512714}{RGB}{255,127,14}
\definecolor{lightgray204}{RGB}{204,204,204}
\definecolor{steelblue31119180}{RGB}{31,119,180}

\begin{axis}[
legend cell align={left},
legend style={
  fill opacity=0.25,
  draw opacity=1,
  text opacity=1,
  draw=lightgray204, font=\tiny
},
tick align=outside,
tick pos=left,
x grid style={darkgray176},
xlabel={Binary Entropy of Bit-Flip Probability},
xmajorgrids,
xmin=0.0348472206030655, xmax=1.04565735704564,
xminorgrids,
xtick style={color=black},
y grid style={darkgray176},
ymajorgrids,
ymin=-0.05, ymax=1.05,
yminorgrids,
ytick style={color=black}
]
\addplot [semithick, steelblue31119180, mark=*, mark size=1, mark options={solid}, only marks]
table {%
0.08079313589591 0.9703
0.11144487505864 0.9561127042046
0.13973571890493 0.94206536090937
0.16625232509369 0.92815628625094
0.19134332888486 0.91438479636586
0.21524030552926 0.90075020739071
0.23810840247231 0.88725183546206
0.26007138410538 0.87388899671648
0.28122546891781 0
0.30164760686926 0
0.32140073796551 0
0.34053729470419 0
0.35910162564855 0
0.37713172616107 0
0.39466050753235 0
0.41171674886984 0
0.42832582509175 0
0.44451027322791 0
0.46029023957512 0
0.47568383748573 0
0.49070743705889 0
0.50537590220308 0
0.51970278650431 0
0.53370049647608 0
0.54738042871008 0
0.5607530859425 0
0.57382817593746 0
0.58661469625295 0
0.59912100732124 0
0.61135489578938 0
0.62332362968932 0
0.63503400671278 0
0.64649239663462 0
0.6577047787442 0
0.66867677499721 0
0.67941367948115 0
0.68992048469278 0
0.70020190504564 0
0.71026239796283 0
0.72010618285608 0
0.7297372582487 0
0.73915941726285 0
0.74837626166101 0
0.7573912146056 0
0.76620753227882 0.40415144444444
0.7748283144864 0.40778599023291
0.7832565143529 0.41123499532515
0.79149494720314 0.41450051131146
0.79954629871195 0.41758458978212
0.80741313239617 0.42048928232741
0.815097896512 0.42321664053761
0.82260293041525 0.42576871600301
0.82993047043483 0.42814756031388
0.83708265530399 0.43035522506052
0.84406153118965 0.43239376183321
0.85086905635506 0.43426522222222
0.85750710548793 0.43597165781785
0.86397747372248 0.43751512021037
0.87028188038109 0.43889766099007
0.87642197245864 0.44012133174722
0.88239932787036 0.44118818407213
0.88821545848187 0.44210026955506
0.89387181293858 0.44285963978629
0.89936977930964 0.44346834635612
0.90471068756045 0.44392844085483
0.90989581186646 0.44424197487269
0.91492637277973 0.44441
0.91980353925865 0.44443756782703
0.92452843057062 0.44432372994407
0.92910211807613 0.4440715379414
0.93352562690249 0.443683043409
0.93779993751426 0.44316029793806
0.94192598718729 0.44250535311796
0.9459046713924 0.44172026053928
0.949736845094 0.4408070717923
0.95342332396981 0.43976783846732
0.9569648855548 0.4386046121546
0.96036227031548 0.70916666666667
0.96361618265685 0.71249008264463
0.96672729186733 0.71567245179063
0.96969623300365 0.71871377410468
0.97252360771962 0.72161404958678
0.97520998504141 0.72437327823692
0.97775590209196 0.7269914600551
0.98016186476698 0.72946859504132
0.9824283483647 0.92478377557949
0.98455579817158 0.92216209118736
0.98654463000556 0.91948018407213
0.9883952307189 0.9167373703704
0.99010795866185 1
0.99168314410887 1
0.99312108964848 1
0.99442207053802 1
0.99558633502432 1
0.9966141046313 1
0.99750557441528 1
0.99826091318876 1
0.9988802637133 1
0.99936374286211 1
0.9997114417528 1
};
\addlegendentry{Probability of Success}
\addplot [semithick, darkorange25512714, mark=*, mark size=1, mark options={solid}, only marks]
table {%
0.08079313589591 0.33333333333333
0.11144487505864 0.33333333333333
0.13973571890493 0.33333333333333
0.16625232509369 0.33333333333333
0.19134332888486 0.33333333333333
0.21524030552926 0.33333333333333
0.23810840247231 0.33333333333333
0.26007138410538 0.33333333333333
0.28122546891781 1
0.30164760686926 1
0.32140073796551 1
0.34053729470419 1
0.35910162564855 1
0.37713172616107 1
0.39466050753235 1
0.41171674886984 1
0.42832582509175 1
0.44451027322791 1
0.46029023957512 1
0.47568383748573 1
0.49070743705889 1
0.50537590220308 1
0.51970278650431 1
0.53370049647608 1
0.54738042871008 1
0.5607530859425 1
0.57382817593746 1
0.58661469625295 1
0.59912100732124 1
0.61135489578938 1
0.62332362968932 1
0.63503400671278 1
0.64649239663462 1
0.6577047787442 1
0.66867677499721 1
0.67941367948115 1
0.68992048469278 1
0.70020190504564 1
0.71026239796283 1
0.72010618285608 1
0.7297372582487 1
0.73915941726285 1
0.74837626166101 1
0.7573912146056 1
0.76620753227882 0.66666666666667
0.7748283144864 0.66666666666667
0.7832565143529 0.66666666666667
0.79149494720314 0.66666666666667
0.79954629871195 0.66666666666667
0.80741313239617 0.66666666666667
0.815097896512 0.66666666666667
0.82260293041525 0.66666666666667
0.82993047043483 0.66666666666667
0.83708265530399 0.66666666666667
0.84406153118965 0.66666666666667
0.85086905635506 0.66666666666667
0.85750710548793 0.66666666666667
0.86397747372248 0.66666666666667
0.87028188038109 0.66666666666667
0.87642197245864 0.66666666666667
0.88239932787036 0.66666666666667
0.88821545848187 0.66666666666667
0.89387181293858 0.66666666666667
0.89936977930964 0.66666666666667
0.90471068756045 0.66666666666667
0.90989581186646 0.66666666666667
0.91492637277973 0.66666666666667
0.91980353925865 0.66666666666667
0.92452843057062 0.66666666666667
0.92910211807613 0.66666666666667
0.93352562690249 0.66666666666667
0.93779993751426 0.66666666666667
0.94192598718729 0.66666666666667
0.9459046713924 0.66666666666667
0.949736845094 0.66666666666667
0.95342332396981 0.66666666666667
0.9569648855548 0.66666666666667
0.96036227031548 1
0.96361618265685 1
0.96672729186733 1
0.96969623300365 1
0.97252360771962 1
0.97520998504141 1
0.97775590209196 1
0.98016186476698 1
0.9824283483647 1
0.98455579817158 1
0.98654463000556 1
0.9883952307189 1
0.99010795866185 1
0.99168314410887 1
0.99312108964848 1
0.99442207053802 1
0.99558633502432 1
0.9966141046313 1
0.99750557441528 1
0.99826091318876 1
0.9988802637133 1
0.99936374286211 1
0.9997114417528 1
};
\addlegendentry{Compression Rate}
\end{axis}

\end{tikzpicture}
            \caption[width=\textwidth]{Probability of success and compression rate for Schumacher compression ($\delta=0.05$) as a function of source entropy for blocklength 8}
            \label{fig:schumacher8}
        \end{subfigure}
        \hfill
        \begin{subfigure}[b]{0.45\textwidth}
            \centering
\begin{tikzpicture}

\definecolor{darkgray176}{RGB}{176,176,176}
\definecolor{darkorange25512714}{RGB}{255,127,14}
\definecolor{lightgray204}{RGB}{204,204,204}
\definecolor{steelblue31119180}{RGB}{31,119,180}

\begin{axis}[
legend cell align={left},
legend style={
  fill opacity=0.25,
  draw opacity=1,
  text opacity=1,
  draw=lightgray204, font=\tiny
},
tick align=outside,
tick pos=left,
x grid style={darkgray176},
xlabel={Binary Entropy of Bit-Flip Probability},
xmajorgrids,
xmin=0.0348472206030655, xmax=1.04565735704564,
xminorgrids,
xtick style={color=black},
y grid style={darkgray176},
ymajorgrids,
ymin=-0.05, ymax=1.05,
yminorgrids,
ytick style={color=black}
]
\addplot [semithick, steelblue31119180, mark=*, mark size=1, mark options={solid}, only marks]
table {%
0.08079313589591 0.960596
0.11144487505864 0.94191587920276
0.13973571890493 0.92350952804297
0.16625232509369 0.90537426831569
0.19134332888486 0.88750742507875
0.21524030552926 0.86990633665279
0.23810840247231 0.85256835462127
0.26007138410538 0.83549084383045
0.28122546891781 0
0.30164760686926 0
0.32140073796551 0
0.34053729470419 0
0.35910162564855 0
0.37713172616107 0
0.39466050753235 0
0.41171674886984 0
0.42832582509175 0
0.44451027322791 0
0.46029023957512 0
0.47568383748573 0
0.49070743705889 0
0.50537590220308 0
0.51970278650431 0
0.53370049647608 0
0.54738042871008 0
0.5607530859425 0
0.57382817593746 0
0.58661469625295 0
0.59912100732124 0
0.61135489578938 0
0.62332362968932 0
0.63503400671278 0
0.64649239663462 0.38438421402016
0.6577047787442 0.3888152
0.66867677499721 0.39293703974073
0.67941367948115 0.39675617296933
0.68992048469278 0.40027882636159
0.70020190504564 0.40351121354205
0.71026239796283 0.406459495084
0.72010618285608 0.40912977850967
0.7297372582487 0.4115281182898
0.73915941726285 0.41366051584411
0.74837626166101 0.41553291954102
0.7573912146056 0.41715122469773
0.76620753227882 0.41852127358025
0.7748283144864 0.4196488554033
0.7832565143529 0.42053970633049
0.79149494720314 0.42119950947407
0.79954629871195 0.42163389489516
0.80741313239617 0.42184843960362
0.815097896512 0.42184866755809
0.82260293041525 0.42164004966601
0.82993047043483 0.42122800378356
0.83708265530399 0.42061789471572
0.84406153118965 0.41981503421624
0.85086905635506 0.41882468098765
0.85750710548793 0.4176520406813
0.86397747372248 0.41630226589714
0.87028188038109 0.41478045618415
0.87642197245864 0.4130916580399
0.88239932787036 0.41124086491087
0.88821545848187 0.40923301719217
0.89387181293858 0.40707300222781
0.89936977930964 0.4047656543105
0.90471068756045 0.68669900436877
0.90989581186646 0.6886199339794
0.91492637277973 0.6903184
0.91980353925865 0.69179625456713
0.92452843057062 0.69305525634306
0.92910211807613 0.69409727051568
0.93352562690249 0.69492416879845
0.93779993751426 0.6955378494305
0.94192598718729 0.69594023717642
0.9459046713924 0.69613328332665
0.949736845094 0.69611896569706
0.95342332396981 0.69589928862919
0.9569648855548 0.69547628299
0.96036227031548 0.6948520061728
0.96361618265685 0.69402854209549
0.96672729186733 0.69300800120212
0.96969623300365 0.69179252046234
0.97252360771962 0.84643499132436
0.97520998504141 0.8492273384884
0.97775590209196 0.8518738171861
0.98016186476698 0.85437491979919
0.9824283483647 0.96824963617643
0.98455579817158 0.96676557166151
0.98654463000556 0.9652300794857
0.9883952307189 0.96364198506173
0.99010795866185 1
0.99168314410887 1
0.99312108964848 1
0.99442207053802 1
0.99558633502432 1
0.9966141046313 1
0.99750557441528 1
0.99826091318876 1
0.9988802637133 1
0.99936374286211 1
0.9997114417528 1
};
\addlegendentry{Probability of Success}
\addplot [semithick, darkorange25512714, mark=*, mark size=1, mark options={solid}, only marks]
table {%
0.08079313589591 0.3
0.11144487505864 0.3
0.13973571890493 0.3
0.16625232509369 0.3
0.19134332888486 0.3
0.21524030552926 0.3
0.23810840247231 0.3
0.26007138410538 0.3
0.28122546891781 1
0.30164760686926 1
0.32140073796551 1
0.34053729470419 1
0.35910162564855 1
0.37713172616107 1
0.39466050753235 1
0.41171674886984 1
0.42832582509175 1
0.44451027322791 1
0.46029023957512 1
0.47568383748573 1
0.49070743705889 1
0.50537590220308 1
0.51970278650431 1
0.53370049647608 1
0.54738042871008 1
0.5607530859425 1
0.57382817593746 1
0.58661469625295 1
0.59912100732124 1
0.61135489578938 1
0.62332362968932 1
0.63503400671278 1
0.64649239663462 1
0.6577047787442 1
0.66867677499721 1
0.67941367948115 1
0.68992048469278 1
0.70020190504564 1
0.71026239796283 1
0.72010618285608 1
0.7297372582487 1
0.73915941726285 1
0.74837626166101 1
0.7573912146056 1
0.76620753227882 1
0.7748283144864 1
0.7832565143529 1
0.79149494720314 1
0.79954629871195 1
0.80741313239617 1
0.815097896512 1
0.82260293041525 1
0.82993047043483 1
0.83708265530399 1
0.84406153118965 1
0.85086905635506 1
0.85750710548793 1
0.86397747372248 1
0.87028188038109 1
0.87642197245864 1
0.88239932787036 1
0.88821545848187 1
0.89387181293858 1
0.89936977930964 1
0.90471068756045 1
0.90989581186646 1
0.91492637277973 1
0.91980353925865 1
0.92452843057062 1
0.92910211807613 1
0.93352562690249 1
0.93779993751426 1
0.94192598718729 1
0.9459046713924 1
0.949736845094 1
0.95342332396981 1
0.9569648855548 1
0.96036227031548 1
0.96361618265685 1
0.96672729186733 1
0.96969623300365 1
0.97252360771962 1
0.97520998504141 1
0.97775590209196 1
0.98016186476698 1
0.9824283483647 1
0.98455579817158 1
0.98654463000556 1
0.9883952307189 1
0.99010795866185 1
0.99168314410887 1
0.99312108964848 1
0.99442207053802 1
0.99558633502432 1
0.9966141046313 1
0.99750557441528 1
0.99826091318876 1
0.9988802637133 1
0.99936374286211 1
0.9997114417528 1
};
\addlegendentry{Compression Rate}
\end{axis}

\end{tikzpicture}
            \caption[width=\textwidth]{Probability of success and compression rate for Schumacher compression ($\delta=0.05$) as a function of source entropy for blocklength 8}
            \label{fig:schumacher16}
        \end{subfigure}
    \end{figure}
    \begin{figure}[ht]\ContinuedFloat
        \begin{subfigure}[b]{0.45\textwidth}
            \centering
            \input{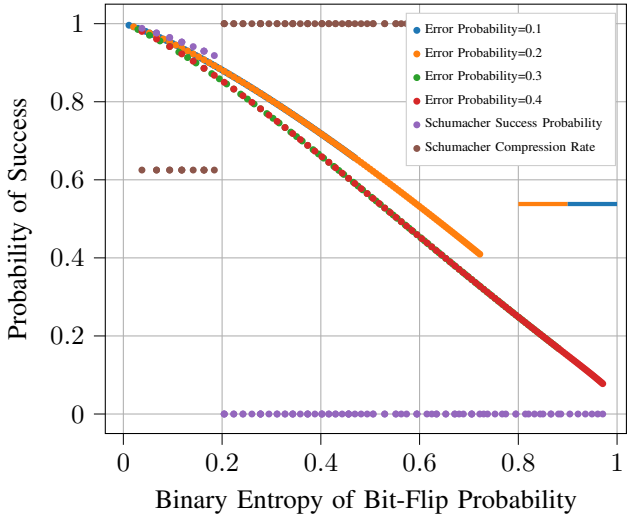}
            \caption[width=0.5\textwidth]{Comparison of Schumacher's protocol ($\delta=0.05$) with ours for blocklength 8. Horizontal lines indicate compression rates.}
            \label{fig:combined8}
        \end{subfigure}
        \hfill
        \begin{subfigure}[b]{0.45\textwidth}
            \centering
            \input{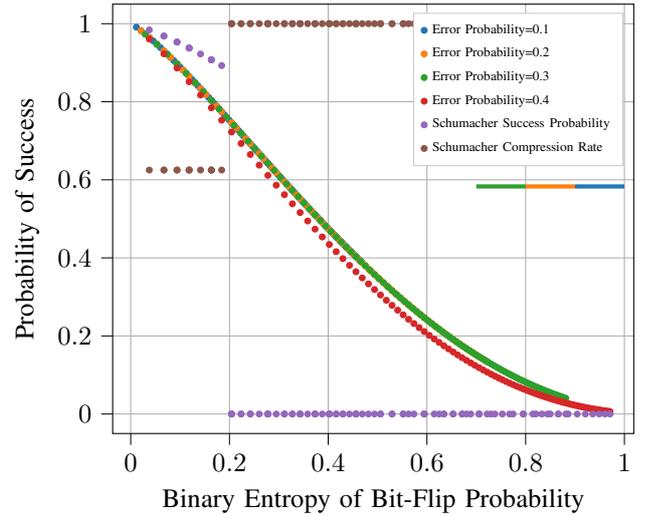}
            \caption[width=0.5\textwidth]{Comparison of Schumacher's protocol ($\delta=0.05$) with ours for blocklength 16. Horizontal lines indicate compression rates.}
            \label{fig:combined16}
        \end{subfigure}
        \caption{Numerical results for our compression protocol and performance comparison with Schumacher compression}
    \end{figure}

\section{Discussion}
\label{Discussion}
In this paper, we consider the problem of quantum state compression and propose an efficient solution using polar codes.
We provide an efficient quantum successive cancellation decoding algorithm based on the factor graph of polar codes.
This provides low-complexity compression and decompression protocols that allow Alice to reliably transmit a multi-qubit quantum state to Bob with a rate approaching the source entropy.
Since our algorithm is based on lifting a classical message-passing decoder to operate on a quantum superposition, the analysis only depends on the classical performance of the polar code.
Thus, we can achieve the Schumacher compression limit $S(\rho)$ using this protocol. Moreover, our proof of proposition (\ref{asymptotically successful}), while nominally specific to polar codes, only uses the fact that polar codes are capacity-achieving in the limit of large block lengths. Thus, any capacity-achieving classical error-correcting code can be used for lossless quantum compression by embedding the corresponding syndrome source coding protocol into the quantum domain.
We have also implemented our algorithm for arbitrary length $N=2^\nu$ in Python. The code can be found on GitHub in the repository:
\begin{center}
\url{https://github.com/Aviemathelec1995/QSCpolar}.
\end{center}


\clearpage

\printbibliography
\newpage
\onecolumn
\begin{appendices}

\section{Efficient Decoding to Recover Compressed Quantum State using Polar Code}
\label{BP decoding}
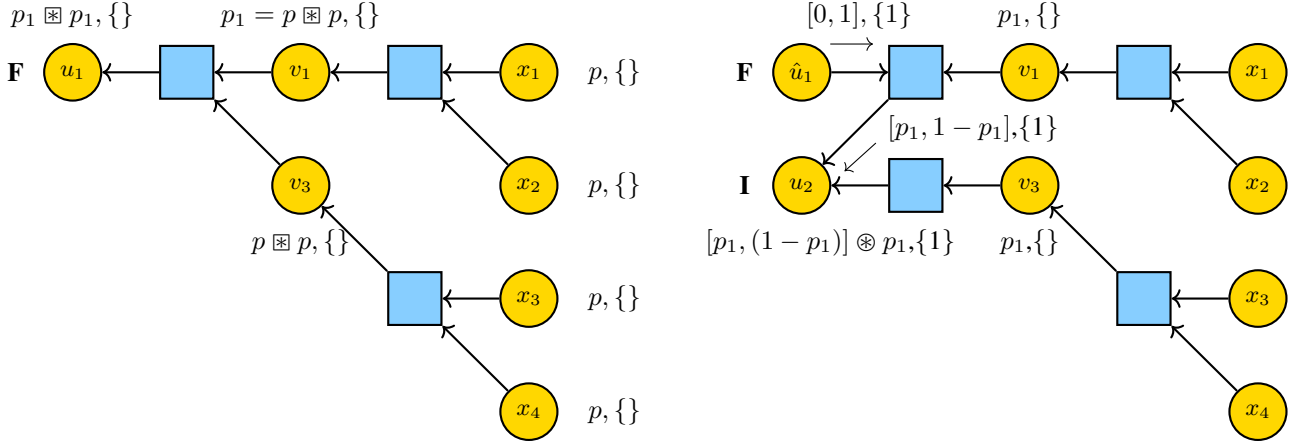
\begin{figure}[h]
    \centering
    \scalebox{1}{\begin{tikzpicture}%
    [scale=0.75,var/.style={font=\small,fill=bitcolor,draw,circle,thick,minimum size=7.5mm},%
    factor/.style={font=\small,fill=checkcolor,draw,rectangle,thick,minimum size=7mm},%
    weight/.style={font=\small}]
    
    \node (u1) [var] at (-0.5,0) {$u_1$} node at (-1.5,0) {\textbf{F}}node at (-0.5,1) {$
    p_1\cnop p_1,\{\}$};
    \node (f1) [factor] at (1.5,0) {};
    \node (v1) [var] at (3.5,0) {$v_1$}
    node at (3.5,1) {$p_1=
    p\cnop p,\{\}$};
    \node (v2) [var] at (3.5,-2) {$v_3$} node at (3.5,-3) {$
    p\cnop p,\{\}$};
    \node (g1) [factor] at (5.5,0) {};
    \node (g3) [factor] at (5.5,-4) {};
    \node (x1) [var] at (7.5,0) {$x_1$} node at (9,0) {$p,\{\}$};
    \node (x2) [var] at (7.5,-2) {$x_2$} node at (9,-2) {$p,\{\}$};
    \node (x3) [var] at (7.5,-4) {$x_3$}
    node at (9,-4) {$p,\{\}$};
    \node (x4) [var] at (7.5,-6) {$x_4$}
    node at (9,-6) {$p,\{\}$};
    
    \draw[->,thick] (x1) -- (g1);      
    \draw[->,thick] (x2) -- (g1);      
    \draw[->,thick] (x3) -- (g3);      
    \draw[->,thick] (x4) -- (g3);      
    \draw[->,thick] (g1) -- (v1);
    \draw[->,thick] (g3) -- (v2);
    \draw[->,thick] (v1) -- (f1);      
    \draw[->,thick] (v2) -- (f1);      
    \draw[->,thick] (f1) -- (u1);

\end{tikzpicture}\hspace{5mm}\begin{tikzpicture}%
    [scale=0.75,var/.style={font=\small,fill=bitcolor,draw,circle,thick,minimum size=7.5mm},%
    factor/.style={font=\small,fill=checkcolor,draw,rectangle,thick,minimum size=7mm},%
    weight/.style={font=\small}]
    
    \node (u1) [var] at (-0.5,0) {$
\hat{u}_1$}  node at (-1.5,0) {\textbf{F}} node at (0.5,1) {$
    [0,1],\{1\}$};
    \draw[->] (0,0.5)--(0.75,0.5) ;
    \node (u2) [var] at (-0.5,-2) {$u_2$} node at (-1.5,-2) {\textbf{I}} node at (0,-3) {$[p_1,(1-p_1)]\vnop p_1$,\{1\}};
    \node (f1) [factor] at (1.5,0) {}
    node at (2.5,-1) {$[p_1,1-p_1]$,\{1\}} ;
    \draw[->] (0.8,-1.25)--(0.25,-1.75);
    \node (f2) [factor] at (1.5,-2) {};
    \node (v1) [var] at (3.5,0) {$v_1$} node at (3.5,1) {$p_1,\{\}$};
    \node (v2) [var] at (3.5,-2) {$v_3$}
    node at (3.5,-3) {$p_1$,\{\}};
    \node (g1) [factor] at (5.5,0) {};
    \node (g3) [factor] at (5.5,-4) {};
    \node (x1) [var] at (7.5,0) {$x_1$};
    \node (x2) [var] at (7.5,-2) {$x_2$};
    \node (x3) [var] at (7.5,-4) {$x_3$};
    \node (x4) [var] at (7.5,-6) {$x_4$};
    
    \draw[->,thick] (x1) -- (g1);      
    \draw[->,thick] (x2) -- (g1);      
    \draw[->,thick] (x3) -- (g3);      
    \draw[->,thick] (x4) -- (g3);      
    \draw[->,thick] (g1) -- (v1);
    \draw[->,thick] (g3) -- (v2);
    \draw[->,thick] (v1) -- (f1);      
    \draw[->,thick] (v2) -- (f2);      
    \draw[<-,thick] (f1) -- (u1);
    \draw[->,thick] (f1) -- (u2);
    \draw[->,thick] (f2) -- (u2);

\end{tikzpicture}}
    \vspace{2mm}
    \caption{Stage 1 (left) and 2 (right) of the message passing algorithm to recover the quantum state. '\textbf{F}' refers to frozen qubit and '\textbf{I}' corresponds to info qubit. }
    \label{polar4dec1}
\end{figure}
    
Before diving into the complete description of the general recursive algorithm to recover the quantum state from its compressed form using length $N=2^n$ polar code, we start by describing the necessary building blocks for the algorithm.  The key challenge in the lifted algorithm is to generalize the check node and the bit node combining rules for conditional messages such that the message probabilities are combined coherently in the right order. In Algorithms \ref{cnop-algo} and \ref{vnop-algo}, we discuss the pseudo-code for check-node and bit-node combining operation for conditional messages. We refer to these pseudo-codes as \textbf{CNOP} and \textbf{BNOP} respectively when we describe the complete decoding algorithm. 
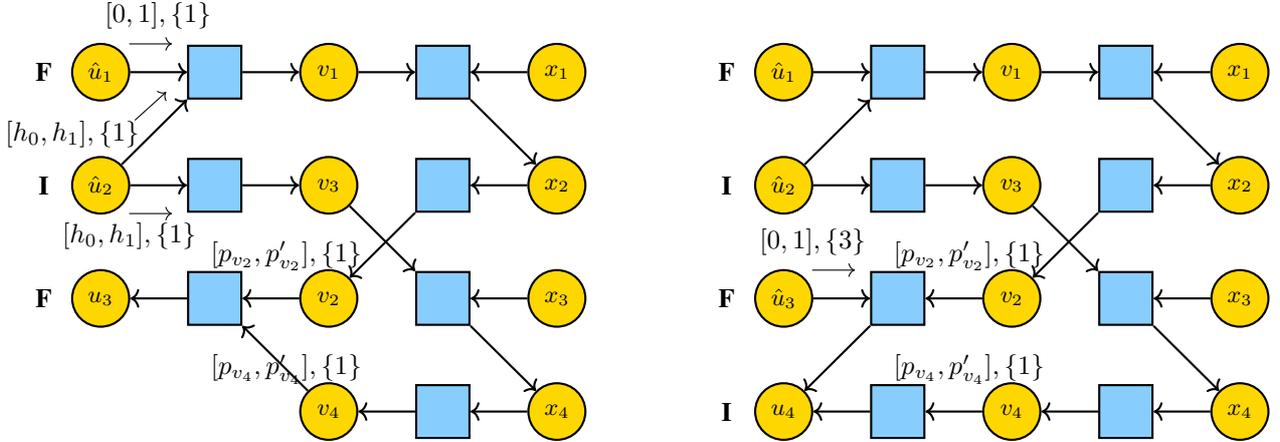
\begin{figure}[h]
    \centering
    \scalebox{1}{\begin{tikzpicture}%
    [scale=0.75,var/.style={font=\small,fill=bitcolor,draw,circle,thick,minimum size=7.5mm},%
    factor/.style={font=\small,fill=checkcolor,draw,rectangle,thick,minimum size=7mm},%
    weight/.style={font=\small}]
    
    \node (u1) [var] at (-0.5,0) {$
\hat{u}_1$}  node at (-1.5,0) {\textbf{F}} node at (0.5,1) {$
    [0,1],\{1\}$};
      \draw[->] (0,0.5)--(0.75,0.5) ;
    \node (u2) [var] at (-0.5,-2) {$\hat{u}_2$} node at (-1.5,-2) {\textbf{I}} node at (-1,-1.1) {$[h_0,h_1],\{1\}$}
    node at (0,-2.9) {$[h_0,h_1],\{1\}$};
    \draw[->] (0.1,-0.85)--(0.65,-0.35) ;
      \draw[->] (0,-2.5)--(0.75,-2.5) ;
    \node (u3) [var] at (-0.5,-4) {$u_3$}node at (-1.5,-4) {\textbf{F}};
    \node (f1) [factor] at (1.5,0) {};
    \node (f2) [factor] at (1.5,-2) {};
    \node (f3) [factor] at (1.5,-4) {};
    \node (v1) [var] at (3.5,0) {$v_1$};
    \node (v2) [var] at (3.5,-2) {$v_3$};
    \node (v3) [var] at (3.5,-4) {$v_2$} node at (2.75,-3.25) {$[p_{v_2},p_{v_2}'],\{1\}$};
    \node (v4) [var] at (3.5,-6) {$v_4$} node at (2.75,-5.25) {$[p_{v_4},p_{v_4}'],\{1\}$};
    \node (g1) [factor] at (5.5,0) {};
    \node (g2) [factor] at (5.5,-2) {};
    \node (g3) [factor] at (5.5,-4) {};
    \node (g4) [factor] at (5.5,-6) {};
    \node (x1) [var] at (7.5,0) {$x_1$};
    \node (x2) [var] at (7.5,-2) {$x_2$};
    \node (x3) [var] at (7.5,-4) {$x_3$};
    \node (x4) [var] at (7.5,-6) {$x_4$};
    
    \draw[->,thick] (x1) -- (g1);      
    \draw[<-,thick] (x2) -- (g1); 
    \draw[->,thick](u2) -- (f1);
    \draw[->,thick] (x2) -- (g2);      
    \draw[->,thick] (x3) -- (g3);      
    \draw[<-,thick] (x4) -- (g3);      
    \draw[->,thick] (x4) -- (g4);
    \draw[<-,thick] (g1) -- (v1);
    \draw[->,thick] (g2) -- (v3);
    \draw[<-,thick] (g3) -- (v2);
    \draw[->,thick] (g4) -- (v4);
    \draw[<-,thick] (v1) -- (f1);      
    \draw[<-,thick] (v2) -- (f2);      
    \draw[->,thick] (v3) -- (f3);      
    \draw[->,thick] (v4) -- (f3);      
    \draw[<-,thick] (f1) -- (u1);
    \draw[<-,thick] (f2) -- (u2);
    \draw[->,thick] (f3) -- (u3);

\end{tikzpicture}\hspace{15mm}\begin{tikzpicture}%
    [scale=0.75,var/.style={font=\small,fill=bitcolor,draw,circle,thick,minimum size=7.5mm},%
    factor/.style={font=\small,fill=checkcolor,draw,rectangle,thick,minimum size=7mm},%
    weight/.style={font=\small}]
    
    \node (u1) [var] at (-0.5,0) {$
\hat{u}_1$} node at (-1.5,0) {\textbf{F}};
    \node (u2) [var] at (-0.5,-2) {$\hat{u}_2$} node at (-1.5,-2) {\textbf{I}};
    \node (u3) [var] at (-0.5,-4) {$\hat{u}_3$}node at (-1.5,-4) {\textbf{F}} node at (0,-3) {$
    [0,1],\{3\}$};
      \draw[->] (0,-3.5)--(0.75,-3.5) ;
    \node (u4) [var] at (-0.5,-6) {$u_4$}node at (-1.5,-6) {\textbf{I}};
    \node (f1) [factor] at (1.5,0) {};
    \node (f2) [factor] at (1.5,-2) {};
    \node (f3) [factor] at (1.5,-4) {};
    \node (f4) [factor] at (1.5,-6) {};
    \node (v1) [var] at (3.5,0) {$v_1$};
    \node (v2) [var] at (3.5,-2) {$v_3$};
    \node (v3) [var] at (3.5,-4) {$v_2$}
    node at (2.75,-3.25) {$[p_{v_2},p_{v_2}'],\{1\}$};
    \node (v4) [var] at (3.5,-6) {$v_4$} node at (2.75,-5.25) {$[p_{v_4},p_{v_4}'],\{1\}$};
    \node (g1) [factor] at (5.5,0) {};
    \node (g2) [factor] at (5.5,-2) {};
    \node (g3) [factor] at (5.5,-4) {};
    \node (g4) [factor] at (5.5,-6) {};
    \node (x1) [var] at (7.5,0) {$x_1$};
    \node (x2) [var] at (7.5,-2) {$x_2$};
    \node (x3) [var] at (7.5,-4) {$x_3$};
    \node (x4) [var] at (7.5,-6) {$x_4$};
    
    \draw[->,thick] (x1) -- (g1);      
    \draw[<-,thick] (x2) -- (g1);  
     \draw[->,thick](u2) -- (f1);
    \draw[->,thick] (x2) -- (g2);      
    \draw[->,thick] (x3) -- (g3);      
    \draw[<-,thick] (x4) -- (g3);      
    \draw[->,thick] (x4) -- (g4);
    \draw[<-,thick] (g1) -- (v1);
    \draw[->,thick] (g2) -- (v3);
    \draw[<-,thick] (g3) -- (v2);
    \draw[->,thick] (g4) -- (v4);
    \draw[<-,thick] (v1) -- (f1);      
    \draw[<-,thick] (v2) -- (f2);      
    \draw[->,thick] (v3) -- (f3);      
    \draw[<-,thick] (u4) -- (f3);      
    \draw[->,thick] (v4) -- (f4);
    \draw[<-,thick] (f1) -- (u1);
    \draw[<-,thick] (f2) -- (u2);
    \draw[<-,thick] (f3) -- (u3);
    \draw[->,thick] (f4) -- (u4);

\end{tikzpicture}}
    \vspace{2mm}
    \caption{Stage 3 (left) and 4 (right) of the message passing algorithm to recover the quantum state. '\textbf{F}' refers to frozen qubit and '\textbf{I}' corresponds to info qubit. }
    \label{polar4dec3}
\end{figure}

Another essential component of the decoding procedure is a conditional unitary which is to be applied to the information qubits based on hard decisions and the states of frozen qubits. The conditional unitary for the information qubit indexed by $i$, which depends on the frozen qubits with index list $L_{F_{i}}$ and conditional message $\{h_{vec},L_{F_{i}}\}$ where $h_{vec}\in \{0,1\}^{2^{|L_{F_{i}}|}}$, is of the form 
\begin{align}
U_{i,L_{F_i}}(h_{vec})=\sum_{f_0f_1..f_{|L_{F_{i}}|}}\ketbra{f_0f_1..f_{|L_{F_{i}}|}}{f_0f_1..f_{|L_{F_{i}}|}}\otimes \sigma_{x}^{h_{f_0f_1..f_{|L_{F_{i}}|}}}
\end{align}
where $f_0f_1..f_{|L_{F_{i}}|}\in \{0,1\}^{|L_{F_{i}}|}$ is the binary representation of the indices of $h_{vec}$ and $h_{f_0f_1..f_{|L_{F_{i}}|}}$ is the corresponding hard decision value.
\begin{algorithm}[h]
\caption{Check-node combining for conditional messages (CNOP)\label{cnop-algo}}
\begin{algorithmic}[1]
      \Require First conditional message $\{p_{vec},L_{F}\}$, second conditional message $\{p_{vec}',L_{F'}\}$.
      \State Compute intersection of lists $L_{F_{\cap}}=L_{F}\cap L_{F'}$\\
      Compute $L_{F_1}=L_{F}-L_{F_{\cap}}$ and $L_{F_2}=L_{F'}-L_{F_{\cap}}$.\\
      Compute the union list $L_{F_{\cnop}}=L_{F_{\cap}}\cup L_{F_1}\cup L_{F_2}$\\
      Reorder the message probabilities in $p_{vec}$ and $p_{vec}'$ with respect to the frozen qubit index lists $L_{F_{\cap}}\cup L_{F_1}$ and $L_{F_{\cap}}\cup L_{F_2}$ respectively.
      \For{each $i$ in $[2^{|L_{F_{\cap}}|}]$}
      \For{each each $j$ in $[2^{|L_{F_1}|}]$ }
      \For{each each each $k$ in $[2^{|L_{F_2}|}]$}\\
    \quad  \qquad Select the $(i\times 2^{|L_{F_1}|}+j)$\textsuperscript{th} element $q$ from $p_{vec}$\\
   \quad   \qquad Select the $(i\times 2^{|L_{F_2}|}+k)$\textsuperscript{th} element $q'$ from $p_{vec}'$\\
  \quad \qquad   Compute $q_{new}=q\cnop q'$\\
\quad    \qquad  Add $q_{new}$ to the list $p^{\cnop}_{vec}$. 
    \EndFor
      \EndFor
      \EndFor\\
    Sort $L_{F_{\cnop}}$ in the ascending order of frozen qubit indices\\
    Reorder message probabilities in $p^{\cnop}_{vec}$ according to $L_{F_{\cnop}}$\\
    \Return  $\{p^{\cnop}_{vec},L_{F_{\cnop}}\}$
    \end{algorithmic}
\end{algorithm}
 
\begin{algorithm}[t]
\caption{Bit-node combining for conditional messages (BNOP) \label{vnop-algo}}
\begin{algorithmic}[1]
      \Require First conditional message $\{p_{vec},L_{F}\}$, second conditional message $\{p_{vec}',L_{F'}\}$.
      \State Compute intersection of lists $L_{F_{\cap}}=L_{F}\cap L_{F'}$\\
      Compute $L_{F_1}=L_{F}-L_{F_I}$ and $L_{F_2}=L_{F'}-L_{F_I}$.\\
      Compute the union list $L_{F_{\vnop}}=L_{F_{\cap}}\cup L_{F_1}\cup L_{F_2}$\\
      Reorder the message probabilities in $p_{vec}$ and $p_{vec}'$ with respect to the frozen qubit index lists $L_{F_{\cap}}\cup L_{F_1}$ and $L_{F_{\cap}}\cup L_{F_2}$ respectively.
      \For{each $i$ in $[2^{|L_{F_{\cap}}|}]$}
      \For{each each $j$ in $[2^{|L_{F_1}|}]$ }
      \For{each each each $k$ in $[2^{|L_{F_2}|}]$}\\
    \quad  \qquad Select the $(i\times 2^{|L_{F_1}|}+j)$\textsuperscript{th} element $q$ from $p_{vec}$\\
   \quad   \qquad Select the $(i\times 2^{|L_{F_2}|}+k)$\textsuperscript{th} element $q'$ from $p_{vec}'$\\
  \quad \qquad   Compute $q_{new}=q\vnop q'$\\
\quad    \qquad  Add $q_{new}$ to the list $p^{\vnop}_{vec}$. 
    \EndFor
      \EndFor
      \EndFor\\
    Sort $L_{F_{\vnop}}$ in the ascending order of frozen qubit indices\\
    Reorder message probabilities in $p^{\vnop}_{vec}$ according to $L_{F_{\vnop}}$\\
    \Return  $\{p^{\vnop}_{vec},L_{F_{\vnop}}\}$
    \end{algorithmic}
\end{algorithm}
In Algorithm.~\ref{polar_compression}, we provide the pseudocode of the lifted algorithm to recover the compressed quantum state using a suitable polar code. The decoder receives frozen qubits as the compressed quantum state. The information qubits are prepared in state $\ketbra{0}$. We denote the quantum system associated with index $i$ by $A_{i}$ and denote its state as $\ketbra{\psi}_{A_1\dots A_{N}}$. Let $\pi_{i,L_{F_i}}$ be the permutation such that 
\begin{align*}
    \pi_{i,L_{F_i}}(k)=\begin{cases}
     L_{F_{i}}(m),  \quad \text{if $k=m\leq|L_{F_{i}}|$}\\
     i, \quad \text{if $k=|L_{F_i}|+1$}.
     \end{cases}
\end{align*}
\begin{algorithm}[h]
     \caption{Lifted Algorithm to recover compressed quantum state using polar code (PolarSyndromeDecode)}\label{polar_compression}
     \begin{algorithmic}[1]
     \Require A quantum state $\ket{\psi}_{A_1\dots A_{N}}$, index set $\mathcal{I}$ containing indices of code bits such that $|\mathcal{I}|=N_\mathcal{I}$; conditional message list $L_{\mathcal{I}}$, code design vector $F \in \{I,f\}^{N_{\mathcal{I}}}$ (either information bit $I$ or frozen $f$).
\If{$N_\mathcal{I}=1$} \Comment{Recurse down to length 1}
 \If{$F(1) = I$} \Comment{If the bit is information bit}
  \qquad \State  Apply bit reversal on $\mathcal{I}(1)$ to find information qubit index $b_i=\text{bit-reversal}(\mathcal{I}(1))$ \\
 \qquad \quad Find conditional message for the information qubit $\{p_{vec},L_{F}\}=L_{\mathcal{I}}(1)$.\\
 \qquad \quad Apply hard decision; $h_{vec}=\text{hard-dec}(p_{vec})$\\
 \qquad \quad Apply conditional unitary; $\ketbra{\psi}_{A_1\dots A_{N}}\gets \tilde{U}_{b_i,L_{F}}(h_{vec})\ketbra{\psi}_{A_1\dots A_{N}}$.\\
  \qquad \Return quantum state $\ket{\psi}_{A_1\dots A_{N}}$, conditional message $\{h_{vec},L_{F}\}$
  \Else
   \qquad \State  Apply bit reversal on $\mathcal{I}(1)$ to find information qubit index $b_i=\text{bit-reversal}(\mathcal{I}(1))$\\
\qquad   \Return  quantum state $\ketbra{\psi}_{A_1\dots A_{N}}$, message from frozen bit $\{\{0,1\},b_{i}\}$
 \EndIf
\EndIf\\
Define $\mathcal{I}_{odd} = \{ \mathcal{I}(2i-1) \, | \, i\in [N_{\mathcal{I}}/2] \}$.\\
Create an empty conditional message list $L_{\mathcal{I}_{odd}}'$.
\For{each odd $i$ in $[N_\mathcal{I}]$}\\
\quad $L_{\mathcal{I}_{odd}}'(i)\gets$ CNOP$\left(L_{\mathcal{I}}(i),L_{\mathcal{I}}(i+1)\right)$.
\EndFor\\
$\ketbra{\psi}_{A_1\dots A_{N}},L_{\mathcal{I}_{odd}}^{new}\gets$ PolarSyndromeDecode$\left(\ketbra{\psi}_{A_1\dots A_{N}},\mathcal{I}_{odd},L_{\mathcal{I}_{odd}}',F([N_\mathcal{I}/2])\right)$\\
Create an empty conditional message list $L_{\mathcal{I}_{odd}}^{\cnop}$.
\For{each odd $i$ in $[N_\mathcal{I}]$}\\
\quad $L_{\mathcal{I}_{odd}}^{\cnop}(i)\gets$ CNOP$\left(L_{\mathcal{I}}(i),L_{\mathcal{I}_{odd}}^{new}(\lfloor i/2\rfloor+1)\right)$.
\EndFor\\
Define $\mathcal{I}_{even} = \{ \mathcal{I}(2i) \, | \, i\in [N_{\mathcal{I}}/2] \}$.\\
Create an empty conditional message list $L_{\mathcal{I}_{even}}'$.
\For{each even $i$ in $[N_\mathcal{I}]$}\\
\quad $L_{\mathcal{I}_{even}}'(i)\gets$ BNOP$\left(L_{\mathcal{I}}(i),L_{\mathcal{I}_{odd}}^{\cnop}(\lfloor i/2\rfloor)\right)$.
\EndFor\\
$\ketbra{\psi}_{A_1\dots A_{N}},L_{\mathcal{I}_{even}}^{new}\gets$ PolarSyndromeDecode$\left(\ketbra{\psi}_{A_1\dots A_{N}},\mathcal{I}_{even},L_{\mathcal{I}_{even}}',F(\{N_\mathcal{I}/2+1,\ldots,N_\mathcal{I}\})\right)$\\
Create an empty conditional message list $L_{\mathcal{I}_{even}}^{\vnop}$.
\For{each even $i$ in $[N_\mathcal{I}]$}\\
\quad $L_{\mathcal{I}_{even}}^{\vnop}(i)\gets$ CNOP$\left(L_{\mathcal{I}_{odd}}^{new}(\lfloor i/2\rfloor),L_{\mathcal{I}_{even}}^{new}(\lfloor i/2\rfloor)\right)$\\
\quad $L_{\mathcal{I}_{even}}^{\vnop}(i+1)\gets$ CNOP$\left(\{\{0\},\{\}\},L_{\mathcal{I}_{even}}^{new}(\lfloor i/2\rfloor)\right)$.
\EndFor\\
\Return The quantum state $\ketbra{\psi}_{A_1\dots A_{N}}$, conditional message list $L_{\mathcal{I}_{even}}^{\vnop}$
\end{algorithmic}
 \end{algorithm}

Let $Sw^{(N)}_{\pi_{i,L_{F_i}}}$ be a unitary which swaps the states in the quantum systems according to the permutation $ \pi_{i,L_{F_i}}$. We construct the unitary $\tilde{U}_{i,L_{F_i}}(h_{vec})$ which has the form 
\begin{align*}
    \tilde{U}_{i,L_{F_i}}(h_{vec})= Sw^{(N)}_{\pi_{i,L_{F_i}}} \left(U_{i,L_{F_{i}}}(h_{vec})\otimes I_{N-|L_{F_i}|-1}\right)Sw^{(N),\dagger}_{\pi_{i,L_{F_i}}}.
\end{align*}
The decoder uses the unitary $\tilde{U}_{i,L_{F_i}}(h_{vec})$ to decide whether or not to flip the information qubit in the quantum system $A_{i}$ on the state $\ketbra{\psi}_{A_1\dots A_{N}}$ based on the hard decision $h_{vec}$ obtained from message passing through factor graph where the messages are conditioned on the frozen qubits indexed by the list $L_{F_i}$. 
\end{appendices}
\end{document}